\documentclass[11pt]{article}
\usepackage{latexsym}
\usepackage{amssymb}
\usepackage{times}
\usepackage{fullpage}
\usepackage{amsmath,amsfonts,amsthm}

\newcommand{\ble}{\begin{lemma}}
\newcommand{\ele}{\end{lemma}}

\newtheorem{lemma}{Lemma}[section]

\newtheorem{theorem}[lemma]{Theorem}

\newtheorem{definition}[lemma]{Definition}

\newtheorem{fact}[lemma]{Fact}

\newcommand{\beao}{\begin{eqnarray*}}
\newcommand{\eeao}{\end{eqnarray*}\noindent}
\newcommand{\beam}{\begin{eqnarray}}
\newcommand{\eeam}{\end{eqnarray}\noindent}

\newcommand{\one}{{\bf 1}}

\newcommand{\RR}{{\mathbb{R}}}

\newcommand{\eps}{\epsilon}
\DeclareMathOperator{\E}{E}

\begin{document}

\title{Rademacher Chaos, Random Eulerian Graphs and The Sparse Johnson-Lindenstrauss Transform}
\author{Vladimir Braverman\footnote{University of California Los Angeles, Computer Science Department. Email: {\tt vova@cs.ucla.edu}.}
\ \ \ \ \
Rafail Ostrovsky\footnote{University of California Los Angeles, Computer Science and Mathematics Departments. Email: {\tt rafail@cs.ucla.edu}.}
\ \ \ \ \
Yuval Rabani\footnote{The Rachel and Selim Benin School of Computer
Science and Engineering, The Hebrew University of Jerusalem,
Jerusalem 91904, Israel. Email: {\tt yrabani@cs.huji.ac.il}.
Research supported by ISF grant 1109-07 and BSF grant 2008059.}\\}

\begin{titlepage}

\maketitle
\begin{abstract}
The celebrated dimension reduction lemma of Johnson and Lindenstrauss has
numerous computational and other applications. Due to its application in
practice, speeding up the computation of a Johnson-Lindenstrauss style
dimension reduction is an important question. Recently, Dasgupta, Kumar,
and Sarlos (STOC 2010) constructed such a transform that uses a sparse
matrix. This is motivated by the desire to speed up the computation when
applied to sparse input vectors, a scenario that comes up in applications.
The sparsity of their construction was further improved by Kane and
Nelson (ArXiv 2010).

We improve the previous bound on the number of non-zero entries per
column of Kane and Nelson from
$O(1/\epsilon \log(1/\delta)\log(k/\delta))$ (where the target dimension
is $k$, the distortion is $1\pm \epsilon$, and the failure probability is $\delta$) to
$$
O\left({1\over\epsilon} \left({\log(1/\delta)\log\log\log(1/\delta)
\over \log\log(1/\delta)}\right)^2\right).
$$

We also improve the amount of randomness needed to generate the matrix.
Our results are obtained by connecting the moments of an order 2 Rademacher chaos
to the combinatorial properties of random Eulerian multigraphs. Estimating the
chance that a random multigraph is composed of a given number of node-disjoint
Eulerian components leads to a new tail bound on the chaos. Our
estimates may be of independent interest, and as this part of the argument is
decoupled from the analysis of the coefficients of the chaos, we believe that
our methods can be useful in the analysis of other chaoses.
\end{abstract}

\thispagestyle{empty}
\end{titlepage}

\section{Introduction}

The celebrated flattening lemma of Johnson and Lindenstrauss~\cite{0539.46017}
has numerous applications in pure mathematics, data analysis, signal processing,
computational linear algebra, and machine learning. Informally, the lemma states
that a random linear transformation mapping $\RR^d$ to $\RR^k$, where
$k = O({1\over \epsilon^2}\log(1/\delta))$, preserves the $L_2$-norm of
any $x\in\RR^d$ up to a factor of $(1\pm \epsilon)$ with probability at least
$1 - \delta$. The original argument uses a projection onto a random linear subspace.
However, it turns out that many simpler transformations work just as
well~\cite{48193,639795,276876,861189,1400129}.
In particular, a $k\times d$ matrix of $\{-1,0,+1\}$ i.i.d. entries, and
in fact any sub-Gaussian i.i.d. entries, works~\cite{861189,1400129}.
What makes the lemma particularly useful
are its linearity and the fact that the target dimension $k$ depends only on $\eps$
and $\delta$ but not on $d$. Alon~\cite{Alon200331} gave a lower bound on $k$
demonstrating that the above upper bound is nearly the best possible.

Due to its application in practice, speeding up the computation of a
Johnson-Lindenstrauss style dimension reduction beyond the trivial
$O(dk)$ arithmetric operations per vector is an important question.
Achlioptas~\cite{861189}, then Matou\v{s}ek~\cite{1400129},
gained constant factors by using a sparse matrix. In their ground
breaking work, Ailon and Chazelle~\cite{1132597} designed a
fast Johnson-Lindenstrauss transform (FJLT) that asymptotically
beats the $O(dk)$ bound. Their approach of first applying a
preconditioner that ``smears'' input vectors to some extent, then
using a structured linear transformation that works well on
smeared vectors, is prevalent in followup work.
Ailon and Liberty~\cite{1347083} gave a better FJLT, whose running
time is $O(d\log k)$ arithmetic operations per input vector. Further
results in this vein were given in~\cite{1429833,1347083}.

Recently, Dasgupta, Kumar, and Sarlos~\cite{1806737} revisited the
question of designing a sparse JL transform. This is motivated by
the desire to speed up the computation when applied using a small
$\eps$ to sparse input vectors, a scenario that comes up in applications.
They construct a random $k\times d$ transformation matrix with
$c = O\left(\frac 1 \epsilon \log(1/\delta)\log^2(k/\delta)\right)$
non-zero entries per column. They use a trivial deterministic
preconditioner $P$ that duplicates each coordinate $c$ times and
rescales. The choice of $c$ governs the sparsity of the matrix.
The novelty of their approach lies in the construction of the
projection matrix, whose
entries are not independent. This allows them to overcome a lower
bound of $\tilde{\Omega}(\epsilon^{-2})$ on the sparsity of a JL
transform matrix with independent entries~\cite{1400129}. They
construct the projection matrix as follows: pick $\zeta \in \{-1, 1\}^d$
and a hash function $h:[cd]\rightarrow [k]$ uniformly at random.
The $k\times cd$ projection matrix $H$ has $H_{i,j} = \zeta_i\one_{h(i)=j}$.
Notice that $H$ has a single non-zero entry per column, and the
entire transformation $HP$ has $c$ non-zero entries per column.
Kane and Nelson~\cite{DBLP:journals/corr/abs-1006-3585} improve the analysis of this scheme.
They show that taking $c = O\left(\frac 1 \epsilon \log(1/\delta)\log(k/\delta)\right)$
is sufficient.

We provide alternative, tighter, analysis of this scheme and show that
it is sufficient to set
$$c = O\left(\frac 1 \eps \left({\log(1/\delta)\log\log\log(1/\delta) \over
\log\log(1/\delta)}\right)^2\right).$$
In both previous
papers, as well as this work, the starting point is the same:
the argument boils down to
analyzing the distribution of an order $2$ Rademacher chaos
$Z = \sum_{1\le i<j\le d} a_{ij}\zeta_i \zeta_j$, where the coefficients
$a_{ij}$ are derived from the hash function $h$ and the projected
vector $x$. In particular, showing that the transform works for a
particular choice of $c$ boils down to proving a tail inequality
bounding the probability that $Z$ deviates from 0.
We prove such a tail inequality by bounding a judiciously chosen
large even moment of $Z$.

Notice that the monomials in the expansion of $\E[Z^{2m}]$
are (sums of) products of terms in the sum defining $Z$.
As each term involves two indices $i,j$, there is a correspondence
between monomials and graphs on $\{1,2,\dots,d\}$. The
non-zero monomials correspond to graphs where all nodes
have even degree, in other words: unions of node-disjoint
Eulerian graphs. The previous papers resorted to existing
measure concentration inequalities. They implicitly related the
moments to the weight of a subset of the monomials where
the graphs are composed of pairs of parallel edges, and thus
used the combinatorial structure only partially. This approach
seems to hit a barrier when
$c = o\left(\frac 1 \epsilon \log^2(1/\delta)\right)$.

In order to overcome this barrier, we fully exploit the combinatorial
structure of the monomial terms in the expansion of $\E[Z^{2m}]$.
In particular, we prove non-trivial bounds on the probability that
a random multigraph is the union of a given number of disjoint
Eulerian components (the difficulty stems from the fact that this
is not a monotone property). These bounds may be of independent
interest. Moreover, our analysis of the combinatorial structure of
the monomials is decoupled from the use of the specific properties
of the coefficients of the chaos that lead to the specific tail inequality
that we get. Therefore, our methods are likely to be useful in the analysis
of other order 2 Rachemacher chaoses.

Kane and Nelson~\cite{DBLP:journals/corr/abs-1006-3585} also reduce the required amount
of randomness, as compared to the original construction of~\cite{1806737}.
Our analysis also further improves slightly the bound on the
randomness needed. We need $O(\log(1/\delta))$-wise independent
vectors, whereas Kane and Nelson use $O(\log(k/\delta))$-wise independent
vectors.

\subsection*{Definitions, Assumptions and Main Results}

Let $0< \delta, \epsilon < 1$ be two parameters. We assume that $\epsilon \le \log^{-2}{(\delta^{-1})}$.
Define $m = O(\log \delta^{-1})$ and $k =
O(\epsilon^{-2}m)$. Define $C = O(\epsilon^{-1} \left({m\over F(m)}\right)^2)$ for some function $F$ such that $F(m) =O({\log{m}\over \log\log(m)})$. Let $H:[d] \mapsto [k]$ be a
random function and let $\zeta \in \{-1,1\}^d$ be a random vector. Both vectors have $O(m)$-wise independent entries.
Let $x \in R^d$ be a fixed vector such that $||x||_2 = 1$ and $||x||_{\infty}
\le C^{-0.5}$.
Define
$$
Z_t = \sum_{i\neq j\in [d]} x_ix_j \zeta_i \zeta_j \one_{\{H(i) = H(j) = t\}}, \ \ \ \
Z = \sum_{t=1}^k Z_t.
$$
We note that for fixed $H$ each variable $Z_t$ can be seen as a particular case of Rademacher chaos.
Rademacher chaos of order $2$ is defined as a random variable of the form $\sum_{i\neq j\in [d]} a_{i,j} \zeta_i \zeta_j$.
Thus, we consider a special case when $a_{i,j} = x_ix_j$. There are many bounds for
Rademacher chaos, such as Bonami inequality \cite{Bonami70ddd7} and others, see e.g., Blei and Janson \cite{springerlink:10.1007/BF02385577}, Hanson and Wright \cite{1971HansonWright},
{Lata{\l}a \cite{MR1686370}.
In particular, they can be applied for each $Z_t$ for fixed $H$. However, there are two issues with applying general
inequality in our setting. First, we might loose precision, when applied directly
to a \emph{random} sum of (defined by $H$ ) of Rademacher chaoses $Z$. Second, we can employ the
structure of $a_{i,j}$ to achieve better bounds.
Our main technical result is a new tail probability inequality for a random sum of Rademacher chaoses of the special form as above.
In particular we prove:

\noindent
\begin{theorem}\label{lm:main}
There exists an absolute constant $\alpha$ such that if $C > \alpha\epsilon^{-1} \left({m\over F(m)}\right)^2$ and $k>\alpha\epsilon^{-2}m$ then:
$E(Z^{2m}) \le (0.1\epsilon)^{2m}$. Further, there exists an absolute constant $\gamma$ such that
$$
P(|Z| \ge \epsilon) \le \gamma\delta.
$$
\end{theorem}

Thus, we give an improvement to Theorem $2$ from \cite{1806737} and Theorem $10$ from \cite{DBLP:journals/corr/abs-1006-3585}.
It is important to emphasize  the difference between our approach and that of \cite{1806737,DBLP:journals/corr/abs-1006-3585}
Both previous works
 first bound $Z_i$ using known tail bounds for a Rademacher chaos  and then take a union bound for summing the error of all $Z_{1\leq i\leq t}$ in order to upper bound $Z$. We, in contrast provide a new tail inequality.

Next, we note that theorem \ref{lm:main} immediately implies the following, by repeating the arguments from \cite{1806737}:

\begin{theorem}\label{lm:main JL}
There exists a universal constant $\gamma$ and a distribution $\mathcal{D}$ over $k\times d$ matrices with real-valued elements such that if $M \sim \mathcal{D}$ then
for any fixed $x\in R^d$ the following is true. First,
$$
P((1-\epsilon)||x||_2 \le ||Mx||_2 \le (1+\epsilon)||x||_2 ) \ge 1-\gamma\delta.
$$
Second, $Mx$ can be computed in time
$$
O({1\over\epsilon} \left({\log(1/\delta)\log\log\log(1/\delta) \over \log\log(1/\delta)}\right)^2||x||_0).
$$
Third, $M$ can be constructed using vectors with $O(\log(1/\delta))$-independent entries.
\end{theorem}

\subsection{An Informal Explanation}

We take a direct approach to the above problem and try to estimate the moments of $Z$ directly.
That is, we write
$$
Z = \sum_{1\le i<j\le d} x_ix_j\zeta_i \zeta_j \left(\sum_{t=1}^k \one_{\{H(i) = H(j) = t\}}\right).
$$
Further, $Z^{2m}$ can be seen as a sum of all possible monomials which can be constructed from $2m$ elements of the form
$x_ix_j\zeta_i \zeta_j \left(\sum_{t=1}^k \one_{\{H(i) = H(j) = t\}}\right)$.
Thus, we group the terms according to certain criteria and estimate the expectation of term inside each group differently.
It turns out that each monomial with positive expectation corresponds to a multigraph with positive and even degrees.
The expectation depends on the number of connected components of such graphs.
That is, we reduce the problem of estimating moments of $Z$ to the question of how many multigraphs can be constructed
for a given subset of vertices $\{1,2,\dots,i\}$ and a given number of connected components $t$.
It is not hard to see that $t\le i/2$ for graphs with even degrees.
Also, note that there is a direct upper bound on the number of such sequences that is $i^{4m}$.

Informally, we employ the following intuitive fact. If the multigraph has a small number of connected components, then
the total probability of such a graph is very small.
On the other hand, if there are many connected components, then the graph should be sparse with $o(i^{2})$ edges and thus better bounds are possible.
The main technical work is to prove that for \emph{any} number of components, the combined influence of probability and sparsity
in fact gives the required bound.

\section{Reduction to Graphs}
Let $S$ be a sequence of pairs $S = \{S_1,\dots, S_{2m}\}$ where $S_i = \{S_{i,1}, S_{i,2}\}$ such that $1\le S_{i,1}< S_{i,2}\le d$.
Define $A$ to be a set of all such sequences.
Define a random variable
$$
R_S = \prod_{i = 1}^{2m} \left(x_{S_{i,1}}x_{S_{i,2}} \zeta_{S_{i,1}} \zeta_{S_{i,2}} \left(\sum_{t=1}^k \one_{\{H({S_{i,1}}) = H({S_{i,2}}) = t\}}\right)\right).
$$

\begin{fact}\label{fct:4}
$
E(Z^{2m}) = 2^{2m}\sum_{S\in A} E(R_S).
$
\end{fact}
\begin{proof}
We can rewrite:
$$
Z = 2\sum_{1\le i<j\le d} x_ix_j\zeta_i \zeta_j \left(\sum_{t=1}^k \one_{\{H(i) = H(j) = t\}}\right).
$$
The fact follows.
\end{proof}

\begin{definition}
Let $G$ be an undirected connected multigraph with $G=(V, E)$ and $V \subseteq [d]$. Define
$WEIGHT(G) = 0$ if $G$ has at least one vertex with an odd degree and otherwise define
$$
WEIGHT(G) = {1\over k^{|V|-1}}\prod_{v\in V} x_v^{deg(v)}.
$$
Let $G$ be an undirected multigraph and let $G_1,\dots, G_t$ be the connected components of $G$. Define
$$
WEIGHT(G) = \prod_{i=1}^t WEIGHT(G_i).
$$
\end{definition}

\begin{definition}
Let $V \subseteq [d]$. Define
$$
SQUARES(V) = \prod_{v\in V} x_v^2.
$$
\end{definition}

\begin{definition}\label{def:g(s)}
Let $S\in A$. Define $G(S)$ to be the following undirected multigraph.
Vertices of the graph are the numbers that appear in the sequence $S$. That is, the set of vertices of $G(S)$ is $\{v\in [d]: \exists i \in [2m], j\in \{1,2\}  S_{i,j} = v\}$. The multiset of edges of $G(S)$ consists of all edges of the form $(S_{i,1}, S_{i,2})$.
\end{definition}

\begin{definition}
Let $G$ be a multigraph with vertices in $[d]$. Define $Ver(G)$ to be the set of all vertices of $G$ with positive degree.
Define $Edg(G)$ to be a multiset of all edges of $G$.
\end{definition}

\begin{lemma}
$$
E(R_S) = WEIGHT(G(S)).
$$
\end{lemma}
\begin{proof}
Definition \ref{def:g(s)} implies that all vertices of $G(S)$ have positive degree.
It follows that $G(S)$ has a vertex $v$ with an odd degree if and only if
$x_v$ has an odd degree in $R_S$.
In this case we can write $R_S$ as $\zeta_vL$ where $L$ is independent of $\zeta$ and thus
$E(R_S) = 0 =  WEIGHT(G(S))$.

Consider the case when $G(S)$ has only vertices with positive and even degree. First, let us assume that $G(S)$ is connected.
$$
E(R_S) = \prod_{v\in V} x_v^{deg(v)} E(\prod_{i = 1}^{2m} \left(\sum_{t=1}^k \one_{\{H({S_{i,1}}) = H({S_{i,2}}) = t\}}\right)).
$$
Since $G(S)$ is connected we have
$$
E(\prod_{i = 1}^{2m} \left(\sum_{t=1}^k \one_{\{H({S_{i,1}}) = H({S_{i,2}}) = t\}}\right)) = E(\sum_{t=1}^k\prod_{v\in V} \one_{H(v) = t}) = {1\over k^{|V|-1}}.
$$
The case when $G$ has more than one connected component is proven by repeating the above arguments for each connected components and
by noting that the random variables that correspond to components are independent.
\end{proof}


\begin{definition}
For $Q \subseteq [d]$ and define $W_{Q,t}$ to be set of all sequences $S$ such that $Ver(G(S))=Q$, such that $G(S)$ has $t$ connected
components and such that all degrees in $G(S)$ are positive and even.
By symmetry, for any $Q \neq Q'$ such that $|Q| = |Q'|$ we have $|W_{Q,t}|=|W_{Q',t}|$.
\end{definition}

\begin{lemma}
Let $S\in W_{[i],t}$. Then
$$
WEIGHT(G(S)) \le {1\over k^{i-t}}{1\over C^{2m-i}}SQUARES( Ver(G(S))).
$$
\end{lemma}
\begin{proof}
By Definition
$$
WEIGHT(G(S)) = {1\over k^{i-t}}\prod_{v\in Ver(G(S))} x_v^{deg(v)}.
$$
Next, note the following. For every $v$ it is true that: $deg(v)\ge 2$ and $x^2_v\le C^{-1}$. Also, $\sum_{v\in V} deg(v) = 4m$. Thus, we conclude:
$$
WEIGHT(G(S)) \le {1\over k^{i-t}}{1\over C^{2m-i}}\prod_{v\in Ver(G(S))} x_v^{2} =
$$
$$
{1\over k^{i-t}}{1\over C^{2m-i}}SQUARES( Ver(G(S))).
$$

\end{proof}

\begin{fact}\label{fct:dfkj}
Let $S \notin \cup_{i=1}^{2m}\cup_{t=1}^{i/2} W_{[i],t}$.
Then $E(R_S) = 0$.
\end{fact}
\begin{proof}
Consider $S \notin \cup_{i=1}^{2m}\cup_{t=1}^{i} W_{[i],t}$.
Then $G(S)$ has at least one node of odd degree. It follows that $E(R_S) = 0$.

Further, we show that $W_{[i],t} = \emptyset$ for $t> i/2$. Indeed, consider $S\in W_{[i],t}$.
It follows that at least one of the connected components of $G(S)$ has exactly
one node. This contradicts the definition of sequences $S$. Thus, $W_{[i],t} = \emptyset$ and the fact follows.
\end{proof}

\begin{lemma}\label{lm:dkfkfq}
$$
E(Z^{2m}) \le 2^{2m}\sum_{i=1}^{2m} {1\over i!}\sum_{t=1}^{i/2}|W_{[i],t}|{1\over k^{i-t}}{1\over C^{2m-i}}.
$$
\end{lemma}
\begin{proof}
$$
E(Z^{2m}) = 2^{2m}\sum_{S\in A} E(R_S) = \ \ \ \ \  \  \ \ \  \ \ \ \ \ \ \ \  \  \ \ \  \ \ \ \text{(By Fact \ref{fct:dfkj})}
$$
$$
2^{2m}\sum_{i=1}^{2m} \sum_{t=1}^{i/2}\sum_{Q \in [d], |Q| = i}\sum_{S\in W_{Q,t}} E(R_S) =
$$
$$
2^{2m}\sum_{i=1}^{2m} \sum_{t=1}^{i/2}\sum_{Q \in [d], |Q| = i}\sum_{S\in W_{Q,t}} WEIGHT(G(S)) \le
$$
$$
2^{2m}\sum_{i=1}^{2m} \sum_{t=1}^{i/2}\sum_{Q \in [d], |Q| = i}\sum_{S\in W_{Q,t}} {1\over k^{i-t}}{1\over C^{2m-i}}SQUARES(G(S)) \le
$$
$$
2^{2m}\sum_{i=1}^{2m} \sum_{t=1}^{i/2}|W_{[i],t}| {1\over k^{i-t}}{1\over C^{2m-i}}\sum_{Q \in [d], |Q| = i}SQUARES(Q) \le
$$
$$
2^{2m}\sum_{i=1}^{2m} \sum_{t=1}^{i/2}|W_{[i],t}| {1\over k^{i-t}}{1\over C^{2m-i}}{1\over i!}(\sum_{j \in [d]}x_j^2)^i =
$$
$$
2^{2m}\sum_{i=1}^{2m} \sum_{t=1}^{i/2}|W_{[i],t}| {1\over k^{i-t}}{1\over C^{2m-i}}{1\over i!}.
$$
\end{proof}

\subsection{Proof of Theorem \ref{lm:main}}

\begin{proof}
Let $\epsilon<m^{-2}$. Then by Lemma \ref{lm:dkfkf} and by Fact \ref{fct:dskjkjsd} there exists an absolute constant $\alpha$ such that if
$C > \alpha\epsilon^{-1} \left({m\over F(m)}\right)^2$ and $k>\alpha\epsilon^{-2}m$ then
the following is true.
For any $1\le i \le 2m$ and for any $1\le t \le i/2$:
$$
|W_{[i],t}| \le (0.01)^{2m}i!\epsilon^{2m}k^{i-t}C^{2m-i}.
$$

Thus, by Lemma \ref{lm:dkfkfq} for sufficiently large $m$:
$$
E(Z^{2m}) \le 2^{2m}\sum_{i=1}^{2m} {1\over i!}\sum_{t=1}^{i/2}|W_{[i],t}|{1\over k^{i-t}}{1\over C^{2m-i}}\le
$$
$$
\epsilon^{2m}m^2(0.02)^{2m} \le (0.1\epsilon)^{2m}.
$$

To show the second claim, note that $P(|Z| \ge \epsilon) \le P(Z^{2m} \ge \epsilon^{2m})$. Also, recall that $m=O(\log(1/\delta))$.
Since $Z^{2m}$ is a non-negative random variable, the second claim follows from Markov inequality and the first claim of the theorem.
\end{proof}

\begin{fact}\label{fct:dskjkjsd}
$$
k^{i-t}C^{2m-i}\ge \alpha^{m}{1\over \epsilon^{2m}}m^{4m+i-5t}\left({1\over F(m)}\right)^{4m-2i}.
$$
\end{fact}
\begin{proof}
Recall that $\epsilon\le m^{-2}$ and that $t\le i/2\le m$.
$$
k^{i-t}C^{2m-i} \ge {1\over \epsilon^{2i-2t}}(\alpha m)^{i-t}\alpha^{2m-i}{1\over \epsilon^{2m-i}}\left({m\over F(m)}\right)^{4m-2i} \ge
$$
$$
\alpha^{m}{1\over \epsilon^{2m}}{1\over \epsilon^{i-2t}}m^{4m-i-t}\left({1\over F(m)}\right)^{4m-2i} \ge
$$
$$
\alpha^{m}{1\over \epsilon^{2m}}m^{4m+i-5t}\left({1\over F(m)}\right)^{4m-2i}.
$$
\end{proof}

In the remainder of our paper we prove the following main technical lemma.
\begin{lemma}\label{lm:dkfkf}
Let $\epsilon<m^{-2}$. There exists an absolute constant $CONST = O(1)$ such that for any $1\le i \le 2m$ and for any $1\le t \le i/2$:
$$
|W_{[i],t}| \le (CONST)^{2m}i!m^{4m+i-5t}\left({1\over F(m)}\right)^{4m-2i}.
$$
\end{lemma}
\begin{proof}
The lemma follows directly from Lemma \ref{lm:small t},  Lemma \ref{lm:medium t} and  Lemma \ref{lm:large t}.
\end{proof}

\section{Bounding $W_{[i],t}$.}

\begin{fact}\label{fct:g}
There exists a constant $\upsilon$ such that for $F(m) \le {\upsilon\log(m) \over \log\log (m)}$ and for any $x>0$:
$$
F(m)^{F(m)} \le {m}^{0.01}.
$$
\end{fact}
\begin{proof}
Follows from the fact that for small constant $\upsilon$:
$$
F(m)\log(F(m)) \le 0.01\log(m).
$$
\end{proof}

\subsection{Small $t$}

\begin{lemma}\label{lm:small t}
Let $t< 0.39i$ and $\epsilon<m^{-2}$. There exists an absolute constant $CONST = O(1)$ such that for any $1\le i \le 2m$ and for any $1\le t \le i/2$:
$$
|W_{[i],t}| \le (CONST)^{2m}i!m^{4m+i-5t}\left({1\over F(m)}\right)^{4m-2i}.
$$
\end{lemma}
\begin{proof}
It follows from the definition of $W_{[i],t}$ that $|W_{[i],t}| \le i^{4m}$. Also note that $i\le 2m$.
Also, $2i-5t>0.05i$. Thus, there exists a constant $\phi$ such that,
$$
{i^{4m} \over i!m^{4m+i-5t}\left({1\over F(m)}\right)^{4m-2i}} \le \phi^{m}{i^{4m-i}\over m^{4m-i}}{F(m)^{4m-2i} \over m^{0.05i}}.
$$

First, consider the case when $i\le {m \over F(m)}$. Then the lemma follows immediately.
Otherwise, for sufficiently large $m$ and for some constant $\psi$:

$$
{i^{4m} \over i!m^{4m+i-5t}\left({1\over F(m)}\right)^{4m-2i}} \le \psi^m{F(m)^{4m} \over m^{0.05i}}.
$$

The lemma follows from Fact \ref{fct:g}.
\end{proof}

\subsection{Some Facts}

\begin{fact}\label{fct:1}
Let $t$ be such that $3t > i$ and let $S\in W_{[i],t}$. Then $G(S)$ has at least $(3t-i)>0$ components with size exactly $2$.
\end{fact}
\begin{proof}
Each component must have at least $2$ nodes. Thus, there are at most $(i-2t)$ components with more than $2$ nodes.
Thus, there are at least $(3t-i)$ components of size exactly $2$.
\end{proof}

\begin{definition}
Define $SPARSE_{u}$ as a set of all sequences $S$ such that $Ver(G(S))=[i]$, $G(S)$ has at least $u$ components of size two
and such that all vertices of $G(S)$ are of even and positive degree.
\end{definition}

\begin{fact}\label{fct:3}
Let $t$ be such that $3t > i$. Then $W_{[i],t} \subseteq SPARSE_{3t-i}$.
\end{fact}
\begin{proof}
Follows directly from Fact \ref{fct:1} and the definitions.
\end{proof}

\begin{definition}\label{def:ffgghhh}
Let $Q$ be a set of size $(3t-i)$ of pairs of distinct numbers from $[i]$. That is $$Q= \cup_{j=1}^{(3t-i)}\{(q_{2j-1}, q_{2j})\}$$ such that $q_{j} \in [i]$ and $q_{j} \neq q_{j'}$ for any $j\neq j'$.
Let $\mathcal{Q}$  be a set of all such possible $Q$.
For $Q\in \mathcal{Q}$, define $CONCRETE(Q)$ to be a set of all sequences $S$
such that $G(S)$ has connected components with the following sets of vertices: $\{q_1, q_2\}, \dots, \{q_{2(3t-i)-1}, q_{2(3t-i)}\}$.
\end{definition}

\begin{fact}\label{fct:sldjkjdsljsdj}
$$
|\mathcal{Q}| \le {i\choose 2(3t-i)}{(2(3t-i))!\over (3t-i)!}. \ \ \ \
$$
\end{fact}

\begin{fact}\label{fct:kjsdkjsd}
$$
SPARSE_{(3t-i)} \subseteq \left(\bigcup_{Q\in \mathcal{Q}} CONCRETE(Q)\right).
$$
\end{fact}

\subsection{Medium $t$}
In the remainder of the paper we consider the case\footnote{in fact we only need $3t > (1+\gamma)i $ for some constant $\gamma$} when $t > 0.39i$.
Denote $z = i-2t$.
In this section we consider the case when $t$ is not very large such that $z^2 > 2(3t-i)$.

\begin{fact}\label{fct:sdkjdssd}
Let $Q$ be an ordered set of size $(3t-i)$ from Definition \ref{def:ffgghhh}. Then there exists an absolute constant $\gamma$ such that
$$
|CONCRETE(Q)| \le (2m)^{2(3t-i)}(\gamma z)^{4(m-3t+i)}.
$$
\end{fact}
\begin{proof}
Let $A' = \cup_{j=1}^{3t-i}\{(q_{2j-1},q_{2j})\}$.
Let $B = [i]\setminus \{q_1,\dots, q_{2(3t-i)}\}$ and $B' = A'\cup\left(\cup_{j,j' \in B, j<j'} (j,j')\right)$.
If $S\in CONCRETE(Q)$ then $S \in B'^{2m}$.
Also,
$$
|B'| \le 2(3t-i)+(i - 2(3t - i))^2 = 2(3t-i)+(3(z))^2 \le 10z^2.
$$
Also, each pair $(q_{2j-1},q_{2j})$ must appear at least twice in the sequence $S$.
We count the number of such sequences as follows. First, we choose the $2(3t-i)$ locations of the appearances for the
pairs $(q_{2j-1},q_{2j})$. For a fixed set of locations, the number of sequence $S$ that agree on these locations
is bounded by $|B'|^{2m-2(3t-i)}$. The total number of different sets of locations is bounded by $(2m)^{2(3t-i)}$.
This in an over-counting, yet it is sufficient for our goals.
Thus, we conclude that there exists an absolute constant $\beta$ such that
$$
|CONCRETE(Q)| \le (2m)^{2(3t-i)}(\beta z)^{4(m-3t+i)}.
$$
\end{proof}

\begin{lemma}\label{lm:medium t}
Let $t > 0.39i$ such that $z^2>2(3t-i)$. There exists an absolute constant $CONST = O(1)$ such that for any $1\le i \le 2m$ and for any $1\le t \le i/2$:
$$
|W_{[i],t}| \le (CONST)^{2m}i!m^{4m+i-5t}\left({1\over F(m)}\right)^{4m-2i}.
$$
\end{lemma}
\begin{proof}
By Fact \ref{fct:sdkjdssd}, Fact \ref{fct:sldjkjdsljsdj} and Fact \ref{fct:kjsdkjsd}, there exists an absolute constant $\beta$ such that

$$
|W_{[i],t}| \le |SPARSE_{(3t-i)}| \le {i\choose 2(3t-i)}{(2(3t-i))!\over (3t-i)!}(2m)^{2(3t-i)}(\beta z)^{4(m-3t+i)}.
$$

%
%
%
%

Further,
$$
{i\choose 2(3t-i)}{(2(3t-i))!\over (3t-i)!}(2m)^{2(3t-i)}z^{4(m-3t+i)} =
$$
$$
{i!\over (3t-i)!(3z)!}(2m)^{2(3t-i)}z^{4m-12t+4i} \le
$$

Note that $3t-i>0.1i$. Thus,
$$
(3t-i)! > \left({(3t-i)\over e}\right)^{3t-i} \ge \left({i\over 10e}\right)^{3t-i}.
$$
Thus, there exists an absolute constant $\gamma$ such that
$$
|W_{[i],t}| \le \gamma^mi^{i-3t}i!(2m)^{2(3t-i)}z^{4m-12t+4i-3z} = \gamma^mi^{i-3t}i!m^{6t-2i}z^{4m-6t+i}.
$$

%

To prove the lemma, we need to estimate the following quantity:
$$
{m^{6t-2i}z^{4m-6t+i}F(m)^{4m-2i} \over i^{3t-i}m^{4m+i-5t}}.
$$

We show that there exists a constant $\phi$ such that:

$$
{m^{6t-2i}z^{4m-6t+i}F(m)^{4m-2i} \over i^{3t-i}m^{4m+i-5t}} = \phi^m.
$$
Rewrite:
$$
{m^{6t-2i}z^{4m-6t+i}F(m)^{4m-2i} \over i^{3t-i}m^{4m+i-5t}} = {z^{4m-6t+i}F(m)^{4m-2i} \over i^{3t-i}m^{4m+2i-9t}m^z}.
$$
We consider the following three cases.
If $z\le {i\over F(m)}$ then\footnote{We stress that this claim is correct for any $1\le i \le 2m$.} there exists a constant $\psi$:
$$
{z^{4m-6t+i}F(m)^{4m-2i} \over i^{3t-i}m^{4m+2i-9t}m^z} \le \psi^m{(F(m)z)^{4m-6t+i}F(m)^{-3i+6t} \over i^{4m+i-6t}m^z} \le \psi^m.
$$

If $i\le m$ and ${i\over F(m)} < z\le {m\over F(m)}$ then there exists a constant $\gamma$:
$$
{z^{4m-6t+i}F(m)^{4m-2i} \over i^{3t-i}m^{4m+2i-9t}m^z} \le \gamma^m\left({F(m)z\over m}\right)^{4m-9t+2i}{F(m)^{9t-4i} \over m^z} {z^{3t-i}\over i^{3t-i}} \le \gamma^m{F(m)^i \over m^z}.
$$

Finally, if $\max({m\over F(m)}, {i\over F(m)}) < z$ then there exists a constant $\beta$:
$$
{z^{4m-6t+i}F(m)^{4m-2i} \over i^{3t-i}m^{4m+2i-9t}m^z} \le \beta^m{F(m)^{4m} \over m^z}.
$$

The lemma follows from Fact \ref{fct:g}.

\end{proof}

\subsection{Large $t$}
In this section we consider $t$ such that $z^2<2(3t-i)$.
The proof of the following fact is identical to Fact \ref{fct:sdkjdssd} if we note that $z^2<2(3t-i) \le i$.
\begin{fact}\label{fct:sdddddkjdssd}
Let $Q$ be an ordered set of size $2(3t-i)$ from Definition \ref{def:ffgghhh}. Then there exists an absolute constant $\beta$ such that:
$$
|CONCRETE(Q)| \le (2m)^{2(3t-i)}(\beta i)^{2(m-3t+i)}.
$$
\end{fact}

\begin{lemma}\label{lm:large t}
Let $t$ be such that $z^2<2(3t-i)$. There exists an absolute constant $CONST = O(1)$ such that for any $1\le i \le 2m$ and for any $1\le t \le i/2$:
$$
|W_{[i],t}| \le (CONST)^{2m}i!m^{4m+i-5t}\left({1\over F(m)}\right)^{4m-2i}.
$$
\end{lemma}
\begin{proof}
By Fact \ref{fct:sdddddkjdssd}, Fact \ref{fct:sldjkjdsljsdj} and Fact \ref{fct:kjsdkjsd} there exists an absolute constant $\beta$ such that:

$$
|W_{[i],t}| \le |SPARSE_{(3t-i)}| \le {i\choose 2(3t-i)}{(2(3t-i))!\over (3t-i)!}(2m)^{2(3t-i)}(\beta i)^{2(m-3t+i)}.
$$
Further, there exists a constant $\gamma$:
$$
{i\choose 2(3t-i)}{(2(3t-i))!\over (3t-i)!}(2m)^{2(3t-i)}i^{2(m-3t+i)} \le \gamma^mi!(2m)^{2(3t-i)}i^{2m-9t+3i}.
$$
Thus,
$$
{|W_{[i],t}| \over i!m^{4m+i-5t}\left({1\over F(m)}\right)^{4m-2i}} \le \gamma^m{(2m)^{2(3t-i)}i^{2m-9t+3i}F(m)^{4m-2i}\over m^{4m+i-5t}} \le
$$
$$
\gamma^m{i^{2m-9t+3i}F(m)^{4m-2i}\over m^{4m+3i-11t}} \le
\gamma^m(i/m)^{2m-9t+3i}{F(m)^{4m-2i}\over m^{2m-i}} \le \psi^{2m}
$$
for a constant $\psi$.
\end{proof}

\bibliographystyle{plain}
\bibliography{Bibliography.JL}

\end{document}